\documentclass[journal]{IEEEtran}
\usepackage{amsmath,amssymb}
\usepackage[dvips]{graphicx}
\usepackage{amsfonts}
\usepackage[mathscr]{eucal}
\usepackage{latexsym}
\usepackage{amsthm}
\usepackage{exscale}
\usepackage[mathscr]{eucal}
\usepackage{bm}
\usepackage[dvipsnames]{color}
\usepackage{cases}
\usepackage{epsfig}
\usepackage[center,small]{caption}
\usepackage{algorithm}
\usepackage{algorithmic}
\usepackage[verbose,nospace,sort]{cite}
\usepackage{tabularx}
\usepackage{multirow}
\usepackage{balance}
\usepackage{url}
\usepackage{graphicx}
\usepackage{subcaption}
\usepackage[export]{adjustbox}

\scrollmode

\newtheorem{lemma}{Lemma}

\IEEEoverridecommandlockouts

\begin{document}
\title{How Much Data is Needed for Channel Knowledge Map Construction?}
\author{Xiaoli~Xu, \emph{Member, IEEE,} Yong Zeng, \emph{Senior Member, IEEE,}
\thanks{X. Xu and Y. Zeng (Corresponding author) are with the National Mobile Communications Research Laboratory, School of Information Science and Engineering, Southeast University, Nanjing 210096, China (email: {xiaolixu, yong\_zeng}@seu.edu.cn). This work was supported by the National Natural Science Foundation of China with grant number 6504009712 and 62071114.}
}

\maketitle
\begin{abstract}
Channel knowledge map (CKM) has been recently proposed to enable environment-aware communications by utilizing historical or simulation generated  wireless channel data.  This paper studies the construction of one particular type of CKM, namely channel gain map (CGM), by using a finite number of measurements or simulation-generated data, with model-based spatial channel prediction. We try to answer the following question: How much data is sufficient for CKM  construction? To this end, we first derive the average mean square error (AMSE) of the channel gain prediction as a function of the sample density of data collection for offline CGM construction, as well as  the number of data points used for online spatial channel gain prediction. To model  the spatial variation of the wireless environment even within each cell, we divide the CGM into subregions and estimate the channel parameters from the local data within each subregion. The parameter estimation error and the channel prediction error based on  estimated channel parameters are derived as functions of the number of data points within the subregion. The analytical results provide useful guide for CGM construction and utilization  by determining the required spatial sample density for offline data collection and  number of  data points to be used for online channel prediction, so that  the desired level of channel prediction accuracy is guaranteed.
\end{abstract}


\section{Introduction}
With the ever-increasing node density and channel dimension in wireless communication networks, the acquisition of real-time channel state information (CSI) purely relying on the conventional pilot based channel estimation becomes costly, and even infeasible in rate-hungry and delay-stringent applications. On the other hand, the abundant location-specific channel data and powerful data-mining capability of wireless networks make it possible to shift from the conventional environment-unaware communication to future environment-aware communication, which facilitates CSI acquisition \cite{CKMTutorial}. One promising technique for realizing environment-aware communication is by leveraging channel knowledge map (CKM), which makes use of the available channel knowledge learnt from the physical environment and/or the wireless measurements \cite{Zeng2021}. As shown in Fig.~\ref{F:CKM}, CKM can be constructed by using location-tagged data from the physical-map assisted simulation or the actual channel measurements by network devices. The site-specific CKM can be stored and managed at the base station (BS). During the CKM utilization stage, the user's physical or virtual  location is used to query the CKM for obtaining a priori channel knowledge between the BS and the user. With such a priori channel knowledge, the required online training overhead for real-time CSI acquisition can be greatly reduced \cite{WuDi2023}.  In addition, the training results could be feedback to the construction algorithm for updating the CKM, so as to reflect the environment dynamics.

\begin{figure}[htb]
\centering
\includegraphics[width=0.5\textwidth]{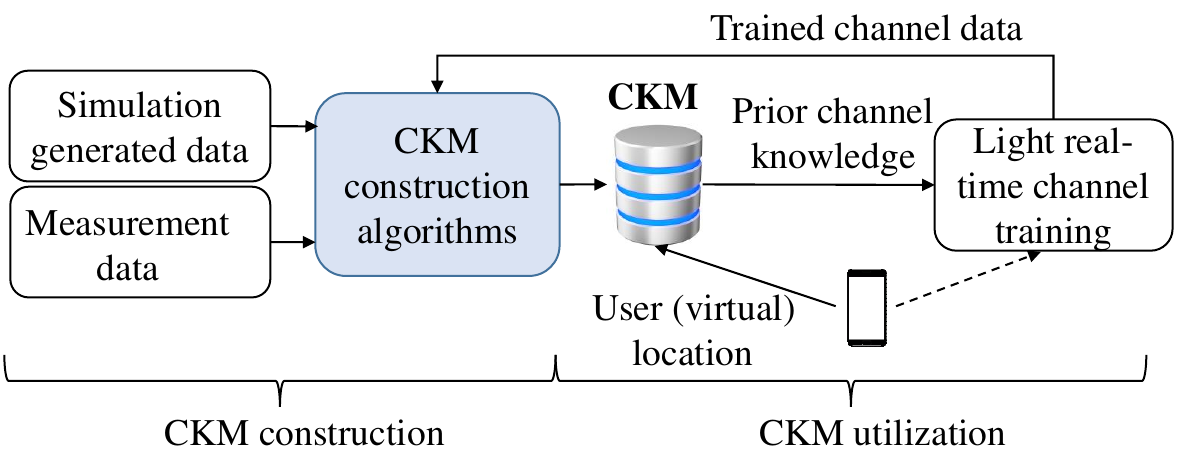}
\caption{An illustration of CKM concept.}
\label{F:CKM}
\end{figure}

 This paper focuses on the construction of a specific type of CKM, termed channel gain map (CGM), which provides channel gain knowledge at arbitrary location within each cell. The most straightforward approach for channel gain prediction is by utilizing stochastic channel models, such as the distance-dependent path loss model. However, such models only use the very high-level environment attributes, such as the environment type: urban, suburban or rural \cite{3GPP}. As a result, the channel gain is only coarsely related to the user location via its  distance from the BS, without considering the specific environment. On the other hand, CGM is a site-specific database constructed based on the actual local wireless environment that is learnt from the physical map or the actual channel measurements. The spatial and temporal correlation of channel gains make it possible to construct  CGM by using a finite number of data samples.

Channel gains of wireless links are typically modelled by  three components, i.e., the path loss, large-scale shadowing and small-scale multi-path fading. The path loss and shadowing are mainly affected by the propagation environment, e.g., the blockage due to building, terrain and plants, which can be considered as relatively static. On the other hand, the multi-path fading de-correlates fast in both the temporal and spatial domain, which is hardly to be predicted in complex environment. Hence, this paper focuses on the \emph{spatial} prediction for CGM construction and aims to determine the required spatial sample density for achieving the desired level of accuracy in channel prediction.   Intuitively, the larger the spatial sample density is, the more likely to find the data points that are close and hence highly correlated with the channel at the target location. However, larger spatial sample density implies high construction and maintenance cost of the CGM. Besides, the number of neighboring data points used for online channel prediction also affects the prediction error. The effect of measurement error and small-scale multi-path fading can be mitigated by including more samples. However, including more samples not only incurs larger computational cost, but also introduce bias if uncorrelated samples are included.

The spatial prediction of channel gain can be achieved in different ways, e.g., the model-based prediction \cite{Malmirchegini12}\cite{Thrane2020}\cite{Liu2023}, data-based prediction \cite{Chowdappa18}\cite{KrigingInference}\cite{VerFunRaj} and the hybrid model and data-based  prediction \cite{Li2022EM}. For model-based prediction, the large-scale path loss at the target location is modeled as a deterministic function of the distance between the target location and the BS. Then, the shadowing at the target location is estimated from the shadowing experienced by the neighboring locations in CGM, which is obtained by subtracting their respective path loss from the measured channel gain. Given the path loss parameters and the spatial correlation model \cite{ShadowingModel}, the minimum mean square error (MMSE) prediction has been proposed in \cite{Malmirchegini12}. The authors in \cite{Malmirchegini12} also investigate the maximum likelihood (ML) and least square (LS) estimation of the channel parameters from the sample measurements. The extension of the spatial prediction algorithms to the spatial-temporal prediction are discussed in \cite{Dionysios15} and \cite{Liao15}, where the temporal predictions are achieved by nonlinear filtering and autoregressive prediction, respectively. Since the path loss and spatial correlation between channel gains are modelled as functions of distances, the accuracy of the model-based spatial prediction is sensitive to localization error \cite{Muppirisetty16}. For data-based prediction, the channel gain is viewed as spatially distributed data and the geostatical tools \cite{GeoSta} are adopted to perform the spatial interpolation, such as the $k$-nearest neighbors (KNN), inverse-distance weight (IDW) and the Kriging Algorithms \cite{Chowdappa18}. Besides, the collected data points can also be used to train a deep neural network (DNN), which is then used to predict the channel gain at the target location \cite{DeepREM}. The performance of the data-based prediction can be improved if the environment information is properly used, e.g., the 3D physical map is used to assist the DNN training process in \cite{Krijestorac2011}. The authors in \cite{Li2022EM} considers the hybrid model and data based channel prediction, where the collected channel  data is first divided into different groups based on the expectation maximization (EM) algorithm, and then  model-based spatial prediction is applied within each data group. However, such data division is only effective when the environment is simple so that the boundary of each group can be clearly identified.

In general, the model-based prediction algorithms  exploit the expert knowledge and enable the estimation under limited data set. However, the performance  of model-based prediction algorithms relies on the accuracy of the channel gain model. On the other hand, the data-driven methods rely on the amount of data available. Compared to model-based prediction, data-based methods  are usually more flexible and may achieve better performance in  complex environment if sufficient data is available. However, data-driven methods are usually not as tractable as the model-based method. To get some insights on the question: ``How much data is needed for CKM construction?", we adopt the classical model-based prediction in this paper and analyze channel prediction performance as a function of the spatial sample density for offline CGM construction, as well as the number of data points used for online channel prediction. Specifically, we consider two types of spatial sample distributions, i.e., the random distribution modelled by homogeneous Poisson point process (PPP) and the grid distribution with certain separation between adjacent samples. When the channel modeling parameters are known, we first derive the average mean square error (AMSE) of the channel gain prediction as a function of the spatial sample density, when only the nearest data point is used for online channel prediction. Then,  the results are extended to the channel prediction by using any number of data points in CKM. The analytical results show that the marginal improvement for AMSE decreases when more samples are included, which implies that a limited number of data points should be sufficient for CKM construction.

One of the main challenges for model-based channel prediction is that the wireless environment may vary significantly even within each site, e.g., the users may experience frequent shift between  line-of-sight (LoS) link and non-LoS (NLoS) link as they move in urban environment. It was shown in \cite{Karttunen2017} that even for a relatively small area, different channel gain models need to be used for users in different streets. Hence, it is desired to establish the local channel gain models for CGM construction. Given the sample density, the expected number of samples within the local region reduces with the region size. The limited number of samples may lead to large error in channel parameter estimation. Since the channel prediction is more sensitive to the path loss parameters, we derive the estimation error of the path loss exponent and path loss intercept as functions of the number of samples within the region. We also point out that the estimation error in channel parameters may not lead to  error in channel prediction, since the correlated shadowing may be counted in the path loss intercept in a small region. Finally, the analytical results are verified by extensive simulations and a case study with actual CGM construction from sparse samples is presented to validate the proposed CGM construction method.

The rest of this paper is organized as follows. The CGM construction problem is formulated in Section~\ref{sec:model}. Section~\ref{sec:prediction} presents the CGM-based spatial channel prediction methods and analyzes the achievable AMSE when the model parameters are known. Section~\ref{sec:estiPara} considers the channel parameter estimation and spatial prediction based on local measurements within a small region. The analytical expressions are verified by the numerical results in Section~\ref{sec:numerical}, and finally Section~\ref{sec:con} draws the conclusion of this paper and summarizes the key results from this study.

\section{System Model}\label{sec:model}
We consider the BS of a macro cell that wants to construct a CGM to enable environment-aware communications. The CGM stores the channel gains at some limited sample locations, which are used to infer the channel gains at any location within the cell, based on the spatial predictability of the channel. According to \cite{Zeng2021}, the CGM can be constructed either by  channel measurements returned by network users at random locations, or by dedicated channel measurement devices at designed locations. Therefore, we consider the following two distributions of the data measurement locations:
\begin{itemize}
\item{\emph{Random Sampling}: Data is collected at random locations, modeled by homogeneous PPP with density $\lambda$.}
\item{\emph{Grid Sampling}: Data is collected  at grid points  separated with a certain distance $d$. }
\end{itemize}

 \begin{figure}[h]
\centering
\begin{subfigure}{0.225\textwidth}
\centering
\includegraphics[width=\textwidth]{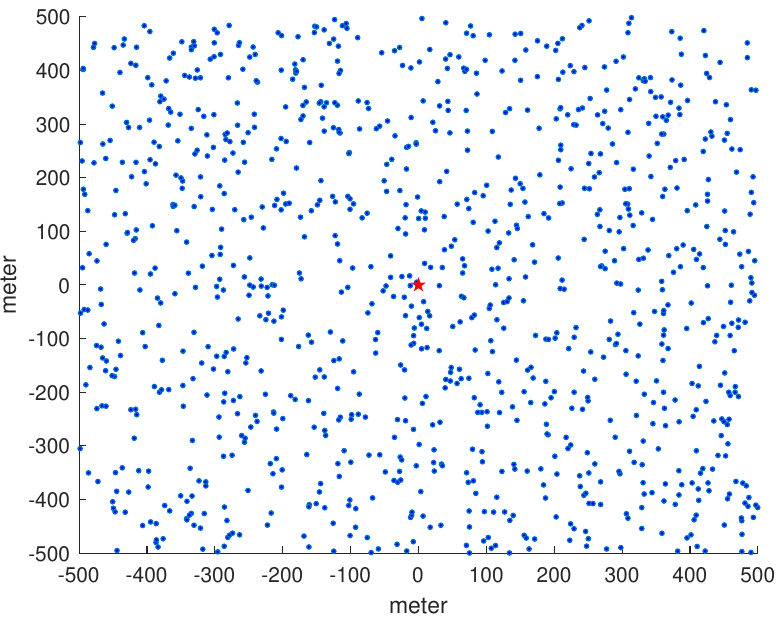}
\caption{Random Sampling}
\end{subfigure}
\begin{subfigure}{0.225\textwidth}
\centering
\includegraphics[width=\textwidth]{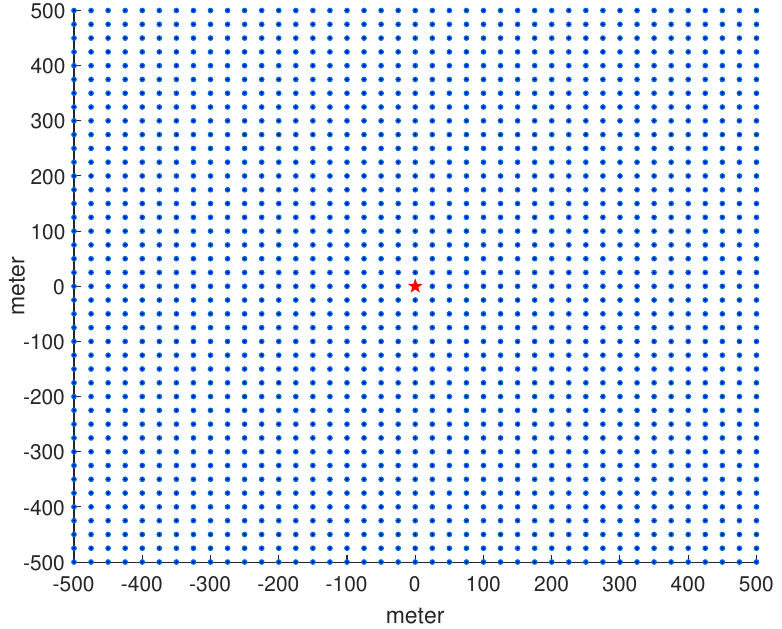}
\caption{Grid Sampling}
\end{subfigure}
\caption{Distribution of data collecting locations for CGM construction in a square area with side length $D=1$km. The red star is the location of the BS. }
\label{F:randomGrid}
\end{figure}

Fig.~\ref{F:randomGrid} presents an illustration example for the random and grid sampling of the data collecting locations. The CGM constructed by the above two methods are referred to as \emph{random CGM} and \emph{grid CGM}, and denoted by $\mathcal{M}_{r}$ and $\mathcal{M}_g$, respectively. The location of the BS is denoted as the origin and the cell is assumed to be a square area, denoted by $\mathcal{A}$, with side length $D$. The number of data points to construct the CGM is $\lambda D^2$ and $(\frac{D}{d})^2$ in random and grid CGM, respectively. According to the spatial predictability of the channel, the channel gain at any arbitrary location $\mathbf{q}\in\mathcal{A}$, denoted by ${\Upsilon}_\mathrm{dB}(\mathbf{q})$, can be estimated from the collected data stored in the CGM. The prediction function is written as
\begin{align}
\tilde{\Upsilon}_\mathrm{dB}(\mathbf{q})=f(\mathbf{q},\mathcal{M}_c), c\in\{r,g\}.
\end{align}

The AMSE of the spatial prediction over all target locations within the area is defined as
\begin{align}
AMSE=\frac{1}{|\mathcal{A}|}\int_{\mathbf q\in \mathcal{A}}\left(\Upsilon_{\mathrm{dB}}(\mathbf{q})-\tilde{\Upsilon}_{\mathrm{dB}}(\mathbf{q})\right)^2d\mathcal A, \label{eq:AMSE}
\end{align}
where $|\mathcal{A}|$ is the size of the area.

Intuitively, the accuracy of the spatial channel prediction improves with the density of the data collecting locations, represented by $\lambda$ and $d$ in $\mathcal{M}_r$ and $\mathcal{M}_g$, respectively. Besides,  the complexity of the estimation algorithm also plays an important role in the spatial prediction, which can be reflected by the number of neighboring data points in CGM used for online channel prediction. In order to answer the question ``how much data is needed for CKM construction?", we  derive the relationship between the AMSE with the data collection density and data used for online prediction.
%\begin{figure}[h]
%\centering
%\subfloat[][]{%
%\adjustbox{valign=b}{
%\includegraphics[width=0.23\textwidth]{Random}}}
%\hspace{0.03cm}
%\subfloat[][]{%
%\adjustbox{valign=b}{
%\includegraphics[width=0.23\textwidth]{Grid}}}
%\caption{Distribution of CGM measurement locations for a square area with side length $D=1$km. }
%\label{F:randomGrid}
%\end{figure}

\section{CGM-based Channel Prediction with Known Parameters}\label{sec:prediction}
In the wireless communication literature, it is well established that the channel gain constitutes of three major components, i.e., the path loss, shadowing and multipath fading. Denote by $\Upsilon(\mathbf{q})$ the channel gain between the BS located at the origin and a receiver at location $\mathbf{q}$. We have $\Upsilon(\mathbf{q})=\Upsilon_{\mathrm{PL}}(\mathbf{q})\Upsilon_{\mathrm{SH}}(\mathbf{q})\Upsilon_{\mathrm{MP}}(\mathbf{q})$, where $\Upsilon_{\mathrm{PL}}(\mathbf{q})$, $\Upsilon_{\mathrm{SH}}(\mathbf{q})$ and $\Upsilon_{\mathrm{ML}}(\mathbf{q})$ represent the impact of path loss, shadowing and multipath fading, respectively. With $\Upsilon_{\mathrm{dB}}(\mathbf q)=10\log_{10}(\Upsilon(\mathbf{q}))$ being the channel gain in dB, we have
\begin{align}
\Upsilon_{\mathrm{dB}}(\mathbf q)=K_{\mathrm{dB}}-10n_{pL}\log_{10}(\|\mathbf q\|)+v(\mathbf q)+\omega(\mathbf q), \label{eq:Ch}
\end{align}
where $K_{\mathrm{dB}}$ accounts the path loss intercept, $n_{PL}$ is the path loss exponent, $v(\mathbf q)=10\log_{10}(\Upsilon_{\mathrm{SH}}(\mathbf{q}))$ is a random variable that models the effect of shadowing, and $\omega(\mathbf q)=10\log_{10}(\Upsilon_{\mathrm{ML}}(\mathbf{q}))$ models the effect of multipath fading. Generally, $v(\mathbf q)$ can be modelled as a zero-mean Gaussian random variable with variance $\alpha$ and spatial correlation function
\begin{align}
\varepsilon(h)=\alpha\exp\left(-\frac{h}{\beta}\right), \label{eq:shadowing}
\end{align}
where $h$ is the distance between the pair of data points and $\beta$ is known as the correlation distance of shadowing \cite{ShadowingModel}. In other words, we have
$\mathbb{E}\{v(\mathbf q_i)v(\mathbf q_j)\}=\varepsilon(\|\mathbf{q}_i-\mathbf{q}_j\|)$.
The multipath effect $\omega(\mathbf{q})$ is also modeled as a zero-mean Gaussian random variable, with variance $\sigma^2$ and it spatially de-correlates fast as compared with the shadowing.

%, we assume the multipath fading is independent at different locations.

%Denote the set of the channel gain values in the CGM by $\mathcal{Y}_{\mathcal Q}$, which are measured at the location set $\mathcal Q$. Given an arbitrary location $\mathbf{q}\notin\mathcal Q$, we want to estimate the channel gain $\Upsilon_{\mathrm{dB}}(\mathbf{q})$ based on the channel model in \eqref{eq:Ch} and the CGM.

\subsection{Spatial Prediction of Channel Gain }
We first assume that the channel modeling parameters in \eqref{eq:Ch}, i.e., $\{n_{\mathrm{PL}},K_{\mathrm{dB}}, \alpha,\beta,\sigma^2\}$, are fixed and known. Then, the estimation of channel gain at location $\mathbf{q}$ is equivalent to estimating the shadowing effect $v(\mathbf{q})$ since the path loss can be computed from the distance to the BS and the multipath fading can be hardly predicted, and hence treated as noise.

Ideally, all data points in the CGM should be used for estimating $\Upsilon_{\mathrm{dB}}(\mathbf{q})$. However, since the channel correlation diminishes with the separation distance and for the purpose of reducing the computational complexity, we assume that only $k$ nearest data collecting locations are included for online channel prediction. Denote the locations of $k$ nearest data point in CGM by  $\mathcal{Q}_k=\{\mathbf q_1,...,\mathbf{q}_k\}$. According to \cite{Malmirchegini12}, the MMSE estimation of $v(\mathbf q)$ is given by
\begin{align}
\tilde{v}(\mathbf q)=\phi_Q^T(\mathbf q)(\mathbf{R}_Q+\sigma^2\mathbf I_{k\times k})^{-1}(\mathbf{y}_Q-K_{\mathrm{dB}}-n_{\mathrm{PL}}\mathbf{h}), \label{eq:shadowEstimation}
\end{align}
where $\phi_Q(\mathbf q)=\left[\varepsilon(\|\mathbf q-\mathbf q_1\|),\cdots,\varepsilon(\|\mathbf q-\mathbf q_k\|) \right]^T$ captures the correlation between the shadowing at the target location $\mathbf{q}$ and that at all the data collecting locations $\mathcal{Q}_k$, $\mathbf{R}_Q$ is the correlation matrix, with the $(i,j)$th component being $[\mathbf{R}_Q]_{ij}=\varepsilon(\|\mathbf{q}_i-\mathbf{q}_j\|)$, $\mathbf I_{k\times k}$ is the identity matrix, $\mathbf{y}_Q=\left[\Upsilon_{\mathrm{dB}}(\mathbf{q}_1),\cdots,\Upsilon_{\mathrm{dB}}(\mathbf{q}_k)\right]^T$ is the measurements stored in the CGM, and $\mathbf h=[-10\log_{10}(\|\mathbf q_{1}\|),\cdots,-10\log_{10}(\|\mathbf q_{k}\|)]^T$ is the distance vector.

%where $\phi_Q(\mathbf q)=\alpha\left[e^{-\frac{\|\mathbf q-\mathbf q_1\|}{\beta}},\cdots,e^{-\frac{\|\mathbf q-\mathbf q_k\|}{\beta}}\right]^T$ captures the correlation between the shadowing at the target location $\mathbf{q}$ and that at all the data collecting locations $\mathcal{Q}_k$, $\mathbf{R}_Q$ is the correlation matrix, with the $(i,j)$th component being $[\mathbf{R}_Q]_{ij}=\varepsilon(\|\mathbf{q}_i-\mathbf{q}_j\|)$, $\mathbf I_{k\times k}$ is the identity matrix, $\mathbf{y}_Q=\left[\Upsilon_{\mathrm{dB}}(\mathbf{q}_1),\cdots,\Upsilon_{\mathrm{dB}}(\mathbf{q}_k)\right]^T$ is the measurements stored in the CGM, and $\mathbf h=[-10\log_{10}(\|\mathbf q_{1}\|),\cdots,-10\log_{10}(\|\mathbf q_{k}\|)]^T$ is the distance vector.

The estimated channel gain $\tilde{\Upsilon}_{\mathrm{dB}}(\mathbf{q})$ is the summation of the path loss and the estimated shadowing $\tilde{v}(\mathbf q)$, i.e.,
\begin{align}
\tilde{\Upsilon}_{\mathrm{dB}}(\mathbf{q})=K_{\mathrm{dB}}-10n_{pL}\log_{10}(\|\mathbf q\|)+\tilde{v}(\mathbf q). \label{eq:estCh}
\end{align}

The complexity of the online channel prediction according to \eqref{eq:estCh} mainly attributes to the inversion  of a $k\times k$ matrix in \eqref{eq:shadowEstimation}, which can be reduced by considering less number of neighboring data points, i.e., using a smaller $k$. Besides, the MSE of the estimation in \eqref{eq:estCh} is given by
\begin{align}
\xi_{\mathrm{dB}}(\mathbf q)&=\left(\Upsilon_{\mathrm{dB}}(\mathbf{q})-\tilde{\Upsilon}_{\mathrm{dB}}(\mathbf{q})\right)^2\nonumber\\
&=\alpha+\sigma^2-\phi_Q^T(\mathbf q)(\mathbf{R}_Q+\sigma^2\mathbf I_{k\times k})^{-1}\phi_Q(\mathbf{q}).\label{eq:MSE}
\end{align}
By substituting \eqref{eq:MSE} into \eqref{eq:AMSE} and taking average with respect to the random location $\mathbf q\in\mathcal A$, we can get the AMSE for the CGM constructed by random and grid sampling, respectively, as elaborated below.

\subsection{AMSE of Random CGM}
We start by considering the special case when only the nearest data point is used, i.e., $k=1$. Denote by $\mathbf{q}_1$ the location of the nearest data point stored in CGM and let $d_{\min}=\|\mathbf{q}-\mathbf{q}_1\|$. Then, the correlation matrix is simplified to  $\mathbf{R}_Q=1$ and the estimated shadowing loss in \eqref{eq:shadowEstimation} reduces to
\begin{align}
\tilde{v}(\mathbf q)=\frac{\alpha}{1+\sigma^2}\exp\left(-\frac{d_{\min}}{\beta}\right)v(\mathbf{q}_1), \label{eq:shadowEstimationk1}
\end{align}
where $v(\mathbf{q}_1)=\Upsilon_{\mathrm{dB}}(\mathbf{q}_1)-K_{\mathrm{dB}}+10n_{\mathrm{PL}}\log_{10}(\|\mathbf{q}_1\|)$ is the shadowing loss at $\mathbf{q}_1$. Furthermore, the MMSE in \eqref{eq:MSE} is simplified to
\begin{align}
\xi_{\mathrm{dB}}(\mathbf q)|_{(k=1)}=\alpha+\sigma^2-\frac{\alpha^2}{\alpha+\sigma^2}\exp\left(-\frac{2d_{\min}}{\beta}\right).\label{eq:MSE_k1}
\end{align}
When the data points are distributed randomly, according to the homogeneous PPP with density $\lambda$, the distribution of $d_{\min}$ is given by the probability density function (p.d.f.) $P_{r}(x)$ \cite{StochasticGoe} with
\begin{align}
P_{r}(x)=2\pi\lambda x\exp\left(-\pi\lambda x^2\right). \label{eq:dmin_random}
\end{align}
By substituting \eqref{eq:MSE_k1} and \eqref{eq:dmin_random} into \eqref{eq:AMSE}, we can obtain the AMSE as a function of the density $\lambda$, which is presented in Lemma~\ref{lem:AMSEk1PPP}.
%\begin{lemma}\label{lem:AMSEk1PPP}
%If the CGM is constructed by measurements at random locations according to the homogeneous PPP with density $\lambda$ and only the nearest measurement is used for the channel gain estimation, the estimation AMSE  is given by
%\begin{align}
%&AMSE=\mathbb{E}[\xi_{\mathrm{dB}}(\mathbf q)|_{(k=1)}]=\alpha+\sigma^2-\frac{\alpha^2}{\alpha+\sigma^2}\nonumber\\
%&\cdot\left(1-\frac{1}{\beta\sqrt{\lambda}}\exp\left(\frac{1}{\pi\lambda\beta^2}\right)\right)\left(1-\Phi\left(\frac{1}{\beta
%\sqrt{\pi\lambda}}\right)\right),\label{eq:AMSEk1PPP}
%\end{align}
%where $\Phi(x)=\frac{2}{\sqrt{\pi}}\int_{0}^xe^{-t^2}dt$ is the Gauss error function.
%\end{lemma}
\begin{lemma}\label{lem:AMSEk1PPP}
If the CGM is constructed by data collection at random locations according to the homogeneous PPP with density $\lambda$ and only the nearest data point is used for online channel  prediction, the AMSE is given by
\begin{align}
&AMSE=\mathbb{E}[\xi_{\mathrm{dB}}(\mathbf q)|_{(k=1)}]=\alpha+\sigma^2-\frac{\alpha^2}{\alpha+\sigma^2}\zeta_r(\lambda),\label{eq:AMSEk1PPP}
\end{align}
where
\begin{align}
\zeta_r(\lambda)&=\int_{0}^{\infty}\exp\left(-\frac{2x}{\beta}\right)P_{r}(x)dx\nonumber\\
&=2\pi\lambda\int_{0}^{\infty}x\exp\left(-\pi\lambda x^2-\frac{2x}{\beta}\right)dx\nonumber\\
&=1-\frac{1}{\beta\sqrt{\lambda}}\exp\left(\frac{1}{\pi\lambda\beta^2}\right)\left(1-\Phi\left(\frac{1}{\beta
\sqrt{\pi\lambda}}\right)\right)\nonumber
\end{align}
and $\Phi(x)=\frac{2}{\sqrt{\pi}}\int_{0}^xe^{-t^2}dt$ is the Gauss error function.
\end{lemma}
\begin{proof}
Please refer to Appendix~\ref{A:PPPk1}.
\end{proof}
When $\lambda\rightarrow 0$, the AMSE in \eqref{eq:AMSEk1PPP} reduces to $(\alpha+\sigma^2)$, which corresponds to the variance of the channel gain prediction using path loss only. When $\lambda\rightarrow\infty$, $\zeta_r(\lambda)\rightarrow 1$ and \eqref{eq:AMSEk1PPP} reduces to $\left(\sigma^2+\frac{\alpha\sigma^2}{\alpha+\sigma^2}\right)$, representing the limit of channel gain spatial prediction. In general, $\zeta_r(\lambda)$ is an increasing function with $\lambda$ and hence the AMSE decreases with the spatial sample density $\lambda$. Besides, the spatial prediction is more effective, i.e., with a larger decaying slope with $\lambda$, when the shadowing variance is larger than the fading power.

Next, we derive the AMSE for general value of $k$, which is rather challenging since it involves the inverse of a random matrix $\mathbf{R}_{Q}$. To gain some insights, we first consider the asymptotical behaviors.  When the CGM density is sufficiently large, i.e., $\lambda\rightarrow\infty$, or the spatial correlation vanishes, i.e., $\beta\rightarrow \infty$, we have $\phi_Q\rightarrow\alpha\mathbf{1}_{k\times 1}$ and $R_Q\rightarrow\mathbf{1}_{k\times k}$, where $\mathbf{1}_{m\times n}$ is a matrix of size $m\times n $ with all the components being 1. Then, the inverse of the matrix can be obtained as
\begin{align}
(\mathbf{R}_Q+\sigma^2\mathbf I_{k\times k})^{-1}&=(\alpha\mathbf{1}_{k\times k}+\sigma^2\mathbf I_{k\times k})^{-1}\nonumber\\
&=\frac{1}{\sigma^2}\mathbf{I}_{k\times k}+\frac{\alpha}{\sigma^2(k\alpha+\sigma^2)}\mathbf{1}_{k\times k}. \label{eq:inverseLim}
\end{align}

By substituting \eqref{eq:inverseLim} and $\Phi_Q$ into \eqref{eq:MSE_k1}, we have
\begin{align}
\lim_{\lambda\rightarrow\infty}\xi_{\mathrm{dB}}(\mathbf q)=\lim_{\beta\rightarrow\infty}\xi_{\mathrm{dB}}(\mathbf q)=\alpha+\sigma^2-\frac{k\alpha^2}{k\alpha+\sigma^2}. \label{eq:MSELim}
\end{align}

Since the limiting behavior of $\xi_{\mathrm{dB}}(\mathbf q)$ no longer depends on the target location, we can obtain AMSE directly. Specifically, for CGM with general density $\lambda$ and prediction algorithm including $k$ neighboring data points, we can obtain the approximate AMSE by assuming that all the $k$ measurements have the minimum distance with the target location, i.e., setting $\phi_Q(\mathbf q)\approx \alpha\exp\left(-\frac{d_{\min}}{\beta}\right), \forall \mathbf q$, and assuming that the measurements are close to each other, i.e., $\mathbf R_Q\approx\mathbf{1}_{k\times k}$. The obtained AMSE is presented in Lemma~\ref{lem:AMSEPPP}.
\begin{lemma}\label{lem:AMSEPPP}
If the CGM is constructed by data collection at random locations according to the homogeneous PPP with density $\lambda$ and $k$ nearest data points are used for online channel prediction, the achieved AMSE is approximated by
\begin{align}
&AMSE\approx \alpha+\sigma^2-\frac{k\alpha^2}{k\alpha+\sigma^2}\zeta_r(\lambda),\label{eq:AMSEkPPP}
\end{align}
where the approximation is tight when $k=1$ or $\lambda\rightarrow\infty$.
\end{lemma}

Similar to \eqref{eq:AMSEk1PPP}, when $\lambda\rightarrow 0$, the sample locations are too far to be correlated, and hence the AMSE is $(\alpha+\sigma^2)$, which is independent of $k$. When $\lambda$ is sufficiently large, the shadowing can be completely known from $k\rightarrow\infty$ samples, and hence the resultant AMSE only contains the fading power $\sigma^2$. In general, for given spatial sample density, the AMSE decays with $k$ with the slope
\begin{align}
\frac{\partial{AMSE}}{\partial k}=-\frac{\alpha^2\sigma^2}{(k\alpha+\sigma^2)^2}\zeta_r(\lambda),
\end{align}
whose absolute value decays with $k$ and increases with $\lambda$. This implies that the marginal improvement of AMSE decreases when $k$ is sufficiently large, and including more data points is effective only when the spatial sample density is large.

\subsection{AMSE of Grid CGM}
When the data collecting locations in CGM are placed in the grid with separation $d$, the distance between a random location $\mathbf q$ with the nearest measurement is distritbuted according to $P_{g}(x)$ with
\begin{align}
P_{g}(x)=\begin{cases}
\frac{2\pi x}{d^2}, & 0\leq x\leq\frac{d}{2}\\
\frac{4x}{d^2}\left(\frac{\pi}{2}-2\arccos\left(\frac{d}{2x}\right)\right), &\frac{d}{2}\leq x\leq\frac{\sqrt{2}d}{2}
\end{cases}\label{eq:dmin_grid}
\end{align}

The derivation of \eqref{eq:dmin_grid} is presented in Appendix~\ref{A:griddmin}. Then, following similar analysis as for random CGM, we can obtain the AMSE for grid CGM with $k=1$, as presented in Lemma~\ref{lem:AMSEK1Grid}.
%\begin{lemma}\label{lem:AMSEK1Grid}
%If the CGM is constructed by measurements at grid points with the minimum separation distance $d$ and only the nearest measurement is used for the channel gain estimation, the estimation AMSE  is given by
%\begin{align}
%&AMSE=\mathbb{E}[\xi_{\mathrm{dB}}(\mathbf q)|_{(k=1)}]=\alpha+\sigma^2-\frac{\alpha^2}{\alpha+\sigma^2}\nonumber\\
%&\cdot\left[\frac{\pi\beta^2}{2d^2}\left(1-\left(1+\frac{\sqrt{2}d}{\beta}\right)e^{-\sqrt{2}d/\beta}\right)-2\zeta\left(\frac{d}{\beta}\right)\right],\label{eq:AMSEk1PPP}
%\end{align}
%where $\zeta(x)=\int_{1}^{\sqrt(2)}ue^{-ux}\left(\arccos\left(\frac{1}{u}\right)\right)du\approx (0.2854-0.0725x+0.0108x^2)e^{-x}$.
%\end{lemma}
\begin{lemma}\label{lem:AMSEK1Grid}
If the CGM is constructed by data collected at grid points with separation $d$ and only the nearest data point is used for online channel gain prediction, the achieved AMSE  is given by
\begin{align}
&AMSE=\mathbb{E}[\xi_{\mathrm{dB}}(\mathbf q)|_{(k=1)}]=\alpha+\sigma^2-\frac{\alpha^2}{\alpha+\sigma^2}\zeta_g(d)
\end{align}
where
\begin{align}
\zeta_g(d)&=\int_{0}^{\frac{d}{2}}\frac{2\pi x}{d^2}e^{-\frac{2x}{\beta}}dx
+\int_{\frac{d}{2}}^{\frac{\sqrt{2}d}{2}}\frac{8 x}{d^2}e^{-\frac{2x}{\beta}}\arccos\left(\frac{d}{2x}\right) dx\nonumber\\
&=\left[\frac{\pi\beta^2}{2d^2}\left(1-\left(1+\frac{\sqrt{2}d}{\beta}\right)e^{-\sqrt{2}d/\beta}\right)-2\Psi\left(\frac{d}{\beta}\right)\right]
\nonumber
\end{align}
and $\Psi(x)=\int_{1}^{\sqrt(2)}ue^{-ux}\left(\arccos\left(\frac{1}{u}\right)\right)du\approx (0.2854-0.0725x+0.0108x^2)e^{-x}$.
\end{lemma}

Following similar analysis as that for random CGM, we can obtain the same asymptotical MSE in \eqref{eq:MSELim} for grid CGM when $d\rightarrow 0$. Furthermore, by using similar approximation, we have
\begin{lemma}\label{lem:AMSEGrid}
If the CGM is constructed by data collected at grid points with separation $d$  separation  $d$ and $k$ nearest data points are used for online channel gain prediction, the achieved AMSE is given by
\begin{align}
&AMSE\approx \alpha+\sigma^2-\frac{k\alpha^2}{k\alpha+\sigma^2}\zeta_g(d),\label{eq:AMSEkGrid}
\end{align}
where the approximation is tight when $k=1$ or $d\rightarrow 0$.
\end{lemma}

Comparing \eqref{eq:AMSEkGrid} with \eqref{eq:AMSEkPPP}, it is observed that the channel prediction performance in random CGM and grid CGM have the same limiting performance at the extremely large or small spatial sample density, i.e., with $\lim_{d\rightarrow 0}\zeta_g(d)=1$ and $\lim_{d\rightarrow \infty}\zeta_g(d)=0$. In general, we have $\zeta_g(d)>\zeta_r(\lambda)$ when they have the same effective density, i.e., with $\lambda=1/d^2$.

The analytical results presented in Lemma~\ref{lem:AMSEPPP} and Lemma~\ref{lem:AMSEGrid} may be used to guide the CGM construction and utilization procedures, by determining  the required sample density for offline CKM construction and the number of data points to be used for online channel prediction, so that the desired level of channel prediction accuracy is guaranteed.

\section{CGM-based Channel Prediction with Unknown Parameters and Spatial Variations}\label{sec:estiPara}
In practice, the channel modeling parameters in the path loss model \eqref{eq:Ch} are usually unknown and they vary across different regions \cite{Karttunen2017}. An illustration example is shown in Fig.~~\ref{F:Map}, where Fig.~~\ref{F:Map}(a) shows a city map and Fig.~~\ref{F:Map}(b) shows the corresponding CGM, generated by the commercial ray tracing software Wireless Insite\footnote{\url{https://www.remcom.
com/wireless-insite-em-propagation-software}}. It is clearly noted that the regions bounded by the three red boxes suffer from different levels of path loss, due to the different levels of LoS link blockage. Hence, different channel gain models should be used for channel spatial prediction. To this end, we divide the area into small regions and achieve the spatial prediction in two steps:
\begin{itemize}
\item{Build the local channel gain model based on the data points within each small region, which is either determined by the distance threshold or building borders.  }
\item{For target location within each region, estimate the channel gain based on the local channel model and the data points in proximity of the target location.}
\end{itemize}

 \begin{figure}[htb]
\centering
\begin{subfigure}{0.35\textwidth}
\centering
\includegraphics[width=\textwidth]{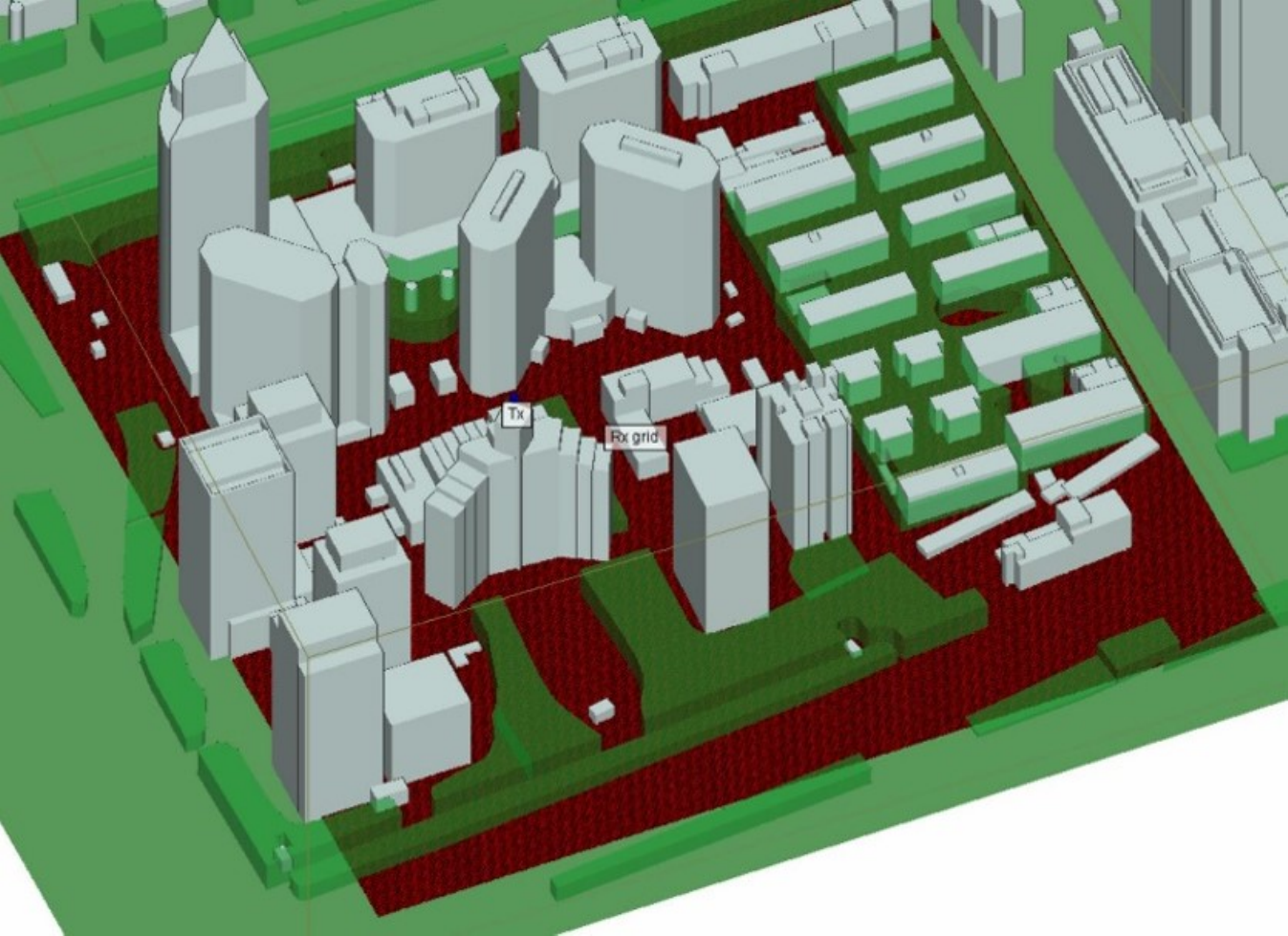}
\caption{City map}
\end{subfigure}
\begin{subfigure}{0.37\textwidth}
\centering
\includegraphics[width=\textwidth]{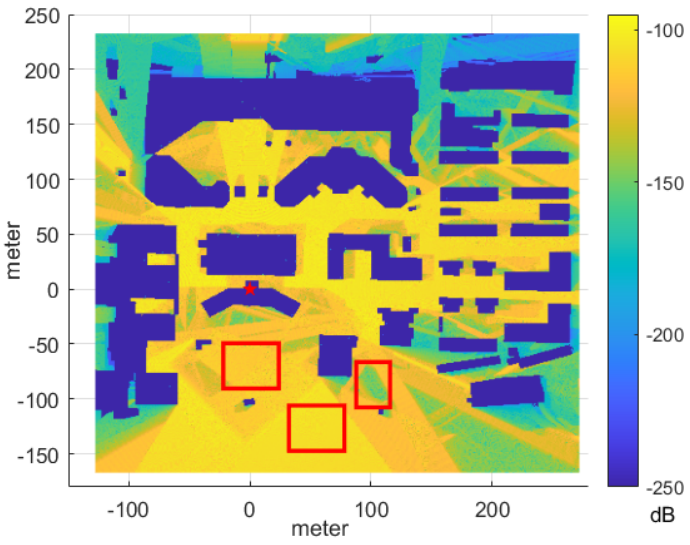}
\caption{Channel gain map}
\end{subfigure}
\caption{An illustration of the channel model variation within a site.}
\label{F:Map}
\end{figure}
In the following, we assume that the region division has been completed and consider the channel modeling parameter estimation for a local region with $N$ measurements in the CGM, located at $\{\mathbf{q}_1,...,\mathbf{q}_N\}$. Furthermore, we analyze the impact of spatial sample density and region size on the channel parameter estimation and the channel prediction.

\subsection{Estimation of Local Channel Modeling Parameters}
Since the effects of shadowing and multipath fading are modeled as the zero-mean random variables, we can obtain the least square (LS) estimation of the path loss intercept and path loss exponent as
\begin{align}
\left[\begin{matrix}\hat{K}_{\mathrm{dB}} & \hat{n}_{\mathrm{PL}}
\end{matrix}\right]=\arg\min\|\mathbf{y}-K_{\mathrm{dB}}-n_{\mathrm{PL}}\mathbf{h}\|^2,
\end{align}
where $\mathbf{y}$ is a vector containing the collected channel gain at all the $N$ sample locations within the  region, $\mathbf{h}=\left[-10\log_{10}(\|\mathbf q_{1}\|),-10\log_{10}(\|\mathbf q_{2}\|),\cdots -10\log_{10}(\|\mathbf q_{N}\|)\right]$ is the vector containing the distance of the $N$ data points to the BS. Let $\mathbf H=[\mathbf{1}_{N\times 1},\mathbf{h}]$, and the LS estimation of the path loss intercept and path loss exponents is given by
\begin{align}
\left[\begin{matrix}\hat{K}_{\mathrm{dB}} \\ \hat{n}_{\mathrm{PL}}
\end{matrix}\right]=\left(\mathbf{H}^T\mathbf{H}\right)^{-1}\mathbf{H}^T\mathbf{y}. \label{eq:LSest}
\end{align}
The covariance matrix of the estimation error vector is
\begin{align}
\mathbf{C}_{LS}&=(\mathbf{H}^T\mathbf{H})^{-1}\mathbf{H}^T\left(\mathbf R_Q+\sigma^2\mathbf{I}_{N\times N}\right)\mathbf{H}(\mathbf{H}^T\mathbf{H})^{-1}, \label{eq:paraMMSE}
\end{align}
where $\mathbf{R}_Q$ is the shadowing correlation matrix for  $N$ data points. The estimation error for  path loss intercept $K_{\mathrm{dB}}$ and path loss exponent $n_{\mathrm{PL}}$ are given by
\begin{align}
\sigma^2_{\hat{K}_{\mathrm{dB}}}=\mathbf C_{LS}(1,1), \sigma^2_{\hat{n}_{\mathrm{PL}}}=\mathbf C_{LS}(2,2). \label{eq:ParaestiError}
\end{align}

In general, the estimation errors in \eqref{eq:ParaestiError} depend on the number of data points and their locations. To gain some insights on the impact of region size and sample density, we assume that the distance of the data collection location to the BS is a random number uniformly distributed within the range $[\delta_{\min},\delta_{\max}]$, i.e., $\|\mathbf{q}\|\sim \mathcal{U}(\delta_{\min},\delta_{\max})$. Then, we can approximate the channel parameter estimation error as a function of the number of data points $N$, and their distribution parameters $(\delta_{\min},\delta_{\max})$, as stated in Lemma~\ref{lem:estiError}.

\begin{lemma}\label{lem:estiError}
Consider a typical region containing $N$ data points in the CGM with random point density $\lambda$ or grid point separation $d$. If the distances of the measurement locations to the BS are uniformly distributed within $[\delta_{\min},\delta_{\max}]$, the LS estimation error of the path loss exponent and path loss intercept in \eqref{eq:paraMMSE} can be approximated by
\begin{align}
\sigma^2_{\hat{K}_{\mathrm{dB}}}&=\frac{\alpha+\sigma^2/c}{N/c}\cdot\frac{\mu^2+\chi}{\chi} \label{eq:KdBError}\\
\sigma^2_{\hat{n}_{\mathrm{PL}}}&=\frac{\alpha+\sigma^2/c}{N/c}\cdot\frac{1}{\chi}\label{eq:nPLerror}
\end{align}
where
\begin{align*}
c&=\begin{cases}\max\{1,\pi\lambda\beta^2\}, & \textnormal{for random CGM}\\
\max\{1,\pi\beta^2/d^2\}, & \textnormal{for grid CGM}\\
\end{cases}\\
\mu&=\frac{10\delta_{\max}\log_{10}(\delta_{\max})-10\delta_{\min}\log_{10}(\delta_{\min})}{\delta_{\max}-\delta_{\min}}-\frac{10}{\ln10}\\
\chi&=\frac{100}{(\ln10)^2}-\frac{100\delta_{\max}\delta_{\min}\log_{10}(\delta_{\max}/\delta_{\min})}{(\delta_{\max}-\delta_{\min})^2}
\end{align*}
\end{lemma}
\begin{proof}
Please refer to Appendix~\ref{A:ParaEstiError}.
\end{proof}

%The closed-form expression in \eqref{eq:KdBError} and $\eqref{eq:nPLerror}$ can well approximate the estimation error of the channel parameters for both the grid CKM and random CKM, as verified by the numerical results in Section~\ref{sec:numerical}.

When the shadowing correlation distance $\beta$ is small, we have $c=1$. Lemma~\ref{lem:estiError} indicates that the estimation error is inversely proportional with the number of data points within the region, which implies that we can improve the parameter estimation accuracy by using a denser CGM. However, when $\beta$ is large, increasing the CGM density will not change the value $N/c$ and hence the improvement over the channel parameter estimation is rather limited.  Besides, the estimation error also depends on the region size, measured by $(\delta_{\max}-\delta_{\min})$, within which the channel parameters are consistent.
\subsection{Estimation of Shadowing and Fading Parameters}
The shadowing and fading parameters can be estimated using LS method or the variogram fitting algorithms discussed in \cite{Cressie85}. This paper adopts the LS estimation for its simplicity and analytical tractability. Specifically, after obtaining the estimated path loss parameters in \eqref{eq:LSest}, the shadowing and fading at each data collecting location can be obtained by subtracting the path loss, i.e.,
\begin{align}
\mathbf{s}=\mathbf{y}-\hat{K}_{\mathrm{dB}}-\hat{n}_{\mathrm{PL}}\mathbf{h},
\end{align}
where $\hat{K}_{\mathrm{dB}}$ and $\hat{n}_{\mathrm{PL}}$ are given in \eqref{eq:LSest}.

Consider a pair of data points at the locations $\mathbf{q}_i$ and $\mathbf{q}_j$, with residual channel gain $\mathbf{s}(i)$ and $\mathbf{s}(j)$, respectively, after subtracting the estimated path loss. The correlation between these residual channel gains is calculated as $\mathbf{s}(i)\mathbf{s}(j)$, which can be used to estimate the shadowing variance and correlation distance in \eqref{eq:shadowing}. Specifically, let $\mathcal{D}=\{d_1,...,d_{|\mathcal D|}\}$ denote the set of all possible distances between a pair sample locations. Further denote by $\mathcal{I}_u=\{(i,j):\|\mathbf{q}_i-\mathbf{q}_j\|=d_u\}$ the set of pairs with distance $d_u$ for $u=1,...,|\mathcal D|$. Then, we can estimate the value of the correlation function at $d_u$ as \begin{align}
\hat{\varepsilon}(d_u)=\frac{1}{|\mathcal{I}_u|}\sum_{(i,j)\in\mathcal{I}_u}\mathbf{s}(i)\mathbf{s}(j).
\end{align}
To reduce complexity and ensure meaningful estimation results, the cardinality of distance vector $|\mathcal{D}|$ is selected such that $\hat{\varepsilon}(d_u)>0$ for $u\leq |\mathcal D|$. Then, the estimation of shadowing parameters can be formulated as
\begin{align}
[\hat{\alpha},\hat{\beta}]=\arg\min\sum_{d_u\in\mathcal{D}}|\mathcal{I}_u|\left(\ln(\alpha e^{-d_u/\beta})-\ln(\hat{\varepsilon}(d_u))\right), \label{eq:alphaEsti}
\end{align}
where $|\mathcal{I}_u|$ is set as the weight measuring the reliability of the experimental value $\hat{\varepsilon}(d_u)$. Problem \eqref{eq:alphaEsti} can be solved easily with the solution given by \cite{Malmirchegini12}
\begin{align}
\left[\begin{matrix}\ln(\hat{\alpha})\\
\frac{1}{\beta}
\end{matrix}\right]=\left(\mathbf{M}^T\mathbf{W}\mathbf{M}\right)^{-1}\mathbf{M}^T\mathbf{W}
\left[\begin{matrix}\ln(\hat{\varepsilon}(d_1)) \\ \vdots \\ \ln(\hat{\varepsilon}(d_{|\mathcal{D}|}))\end{matrix}\right], \label{eq:estShadow}
\end{align}
where $\mathbf{M}=\left[\begin{matrix}1, & \cdots & 1\\ -d_1 & \cdots &-d_{|\mathcal D|}\end{matrix}\right]^T$, $\mathbf{W}=\mathrm{diag}(|\mathcal{I}_1|,\cdots,|\mathcal{I}_{\mathcal{D}}|)$

With $\hat{\alpha}$ obtained in \eqref{eq:estShadow}, the estimated multipath fading variance is given by
\begin{align}
\hat{\sigma}^2=\max\{\mathbf{s}^T\mathbf{s}-\hat{\alpha},0\}. \label{eq:estiSigma}
\end{align}

\subsection{AMSE of Spatial Channel Prediction}
Based on the estimated channel modeling parameters  within the local region, we can perform the spatial prediction of the channel gain within the region as in \eqref{eq:estCh}, by substituting the corresponding estimated channel parameters $\{\hat{n}_{\mathrm{PL}},\hat{K}_{\mathrm{dB}},\hat{\alpha},\hat{\beta},\hat{\sigma}^2\}$. Since the AMSE of channel gain spatial prediction is not sensitive to the estimation error of the correlation distance and shadowing power. Hence, we focus on analyzing the impact of path loss estimation error on the AMSE.

Under the special case when the correlation distance $\beta$ is smaller than the separation between samples, the channel gain estimation reduces to the estimation of path loss, and hence the AMSE is proportional to the path loss parameter estimation error, given in Lemma~\ref{lem:estiError}. Specifically, the channel gain estimation at $\mathbf{q}$ is given by
\begin{align}
\tilde{\Upsilon}_{\mathrm{dB}}(\mathbf q)=\hat{K}_{\mathrm{dB}}-10\hat{n}_{\mathrm{PL}}\log_{10}(\|\mathbf{q}\|))
\end{align}
The MSE of the path loss estimation is given by
\begin{align}
\xi_{\mathrm{dB}}(\mathbf{q})&=\mathbb{E}[(\tilde{\Upsilon}_{\mathrm{dB}}-{\Upsilon}_{\mathrm{dB}})^2]\nonumber\\
&=\alpha+\sigma^2+\sigma_{\hat{K}_{\mathrm{dB}}}^2+(10\log_{10}(\|\mathbf{q}\|)))^2\sigma_{\hat{n}_{\mathrm{PL}}}^2\nonumber\\
&-10\log_{10}(\|\mathbf{q}\|))(\mathbf{C}_{\mathrm{LS}}(2,1)+\mathbf{C}_{\mathrm{LS}}(1,2))
\end{align}
where $\sigma_{\hat{K}_{\mathrm{dB}}}^2$ and $\sigma_{\hat{n}_{\mathrm{PL}}}^2$ are given by \eqref{eq:KdBError} and \eqref{eq:nPLerror}, respectively. Further assuming $\|\mathbf{q}\|\sim \mathcal{U}(\delta_{\min},\delta_{\max})$, we get the AMSE as
\begin{align}
AMSE=\mathbb{E}[\xi_{\mathrm{dB}}(\mathbf{q})]=(\alpha+\sigma^2)\left(1+\frac{2}{N}\right).\label{eq:AMSElowBeta}
\end{align}
which decreases monotonically with the number of data points $N$ within the region.
%\begin{align}
%AMSE=2\frac{\alpha+\sigma^2/c}{N/c}\cdot\frac{\mu^2+\chi}{\chi}
%\end{align}

On the other hand, when $\beta$ is large, the path loss at the target location is calculated from the estimated parameters and the shadowing is estimated by the shadowing experience by the neighboring measurements within the same region. In this case, the AMSE does not necessarily increase with the parameter estimation error, especially when the region is small. This is because the correlated shadowing loss can also be counted into the path loss intercept $\hat{K}_{\mathrm{dB}}$, without affecting the accuracy of channel gain prediction. Hence, deriving the AMSE is challenging. For most of the practical scenarios, we can approximate the AMSE by the case with known channel parameters, presented in Lemma~\ref{lem:AMSEk1PPP} and Lemma~\ref{lem:AMSEPPP}.

\section{Numerical Results}\label{sec:numerical}
The analysis presented in the preceding sections is  verified by the numerical results in this section. We consider the CGM construction in a $D\times D$ square area. Unless otherwise stated, the channel gain is generated according to the model in \eqref{eq:Ch} with channel parameters $n_{\mathrm{PL}}=2.2$, $K_{\mathrm{dB}}=-80$, $\alpha=8$, $\beta=30$ m and $\sigma^2=2$.
%\begin{figure}[htb]
%\centering
%\includegraphics[width=0.45\textwidth]{CGM-Model}
%\caption{The ground truth channel gain generated via the model in \eqref{eq:Ch}. The channel parameters are $n_{\mathrm{PL}}=2.2$, $K_{\mathrm{dB}}=-80$, $\alpha=8$, $\beta=30$ and $\sigma^2=2$.}
%\label{F:CGM_model}
%\end{figure}

First, we verify the analytical AMSE expressions presented in Lemma~\ref{lem:AMSEk1PPP} and Lemma~\ref{lem:AMSEK1Grid} for $k=1$ and various CGM densities in Fig.~\ref{F:AMSEk1}. It is observed that the analysis matches very well with the simulation results, even though an approximation is adopted to obtain the closed-form expression for $\zeta_g(d)$ in Lemma~\ref{lem:AMSEK1Grid}. Besides, the AMSE monotonically decreases with the CGM density. The price paid for denser CGM is the higher construction and storage cost. With the same density, grid CGM slightly outperforms the random CGM since the distance between the random target locations and measurements are guaranteed to be no greater than $d/2$.

\begin{figure}[htb]
\centering
\includegraphics[width=0.37\textwidth]{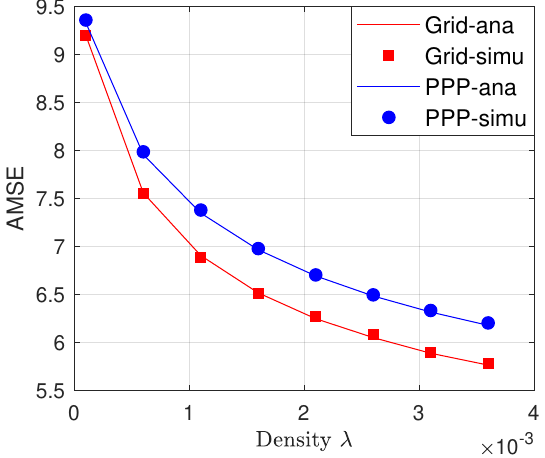}
\caption{The AMSE of channel gain prediction with various CGM densities. For fair comparison, the separation distance in Grid CGM is set as $d=1/\sqrt{\lambda}$.}
\label{F:AMSEk1}
\end{figure}

Under the fixed CGM density, the prediction accuracy can be improved by considering more neighboring data points in CGM for online channel prediction, i.e., with a larger $k$, as indicated by Lemma~\ref{lem:AMSEPPP} and Lemma~\ref{lem:AMSEGrid}. The analytical  results in \eqref{eq:AMSEkPPP}  and \eqref{eq:AMSEkGrid} are verified in Fig.~\ref{F:AMSEk}(a) and Fig.~\ref{F:AMSEk}(b), respectively. As expected, the analytical results are tight when $k=1$ and rather accurate when CGM density is large. It is observed that the nearest data point provides the most information for channel prediction at the target location. With further increase of $k$, the performance improvement is limited, and it becomes almost flat far before $k$ reaches the number of neighboring measurements within the correlation distance, i.e., $\pi\beta^2\lambda$ in random CGM and $\left(\frac{\beta}{d}\right)^2$ in grid CGM. This justifies for the low complexity online channel estimation algorithms that only utilize the limited number of measurements in proximity.

\begin{figure*}[htb]
\centering
\begin{subfigure}{0.35\textwidth}
\centering
\includegraphics[width=\textwidth]{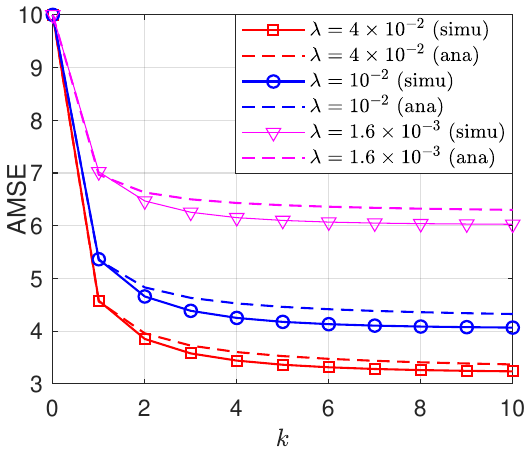}
\caption{Random CGM}
\end{subfigure}
\begin{subfigure}{0.35\textwidth}
\centering
\includegraphics[width=\textwidth]{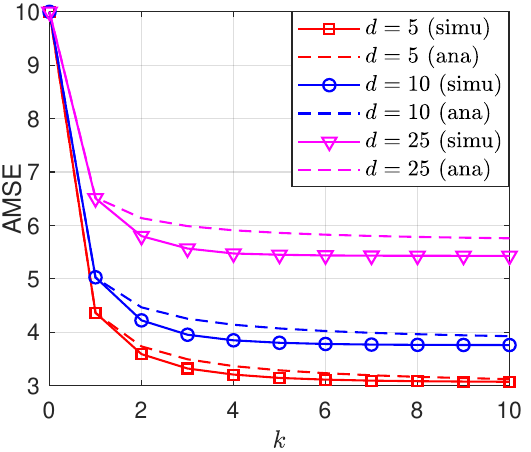}
\caption{Grid CGM}
\end{subfigure}
\caption{The AMSE of channel gain prediction using various number of data points.}
\label{F:AMSEk}
\end{figure*}

Next, we examine the accuracy of the path loss parameter estimation for the scenario with unknown channel modeling parameters. Assume that the channel parameters within a square area of side length $\bar{D}$ is constant, and the border lines of the area is known, so that all the channel gain measurements within that region can be used for channel parameter estimation. Fig.~\ref{F:paraEsti} compares the analytical estimation error of $n_{\mathrm{PL}}$ and $K_{\mathrm{dB}}$, given in \eqref{eq:KdBError} and $\eqref{eq:nPLerror}$, respectively, with the simulation results. The estimation error decreases with the region size within which the channel parameters remain constant. When the shadowing correlation distance is small, i.e., $\beta=1$, the estimation error can be reduced by increasing the map density, as shown in Fig.~\ref{F:paraEsti}(a) and Fig.~\ref{F:paraEsti}(b). However, when $\beta=30$, the reduction of the estimation error with the increase of the data collection density is negligible, as shown in Fig.~\ref{F:paraEsti}(c) and Fig.~\ref{F:paraEsti}(d). The simulation results are consistent with our expectation, as revealed from the analytical results. Besides, similar parameter estimation accuracy is observed for random and grid CGM when they have the same density.
\begin{figure*}[htb]
\centering
\begin{subfigure}{0.33\textwidth}
\centering
\includegraphics[width=\textwidth]{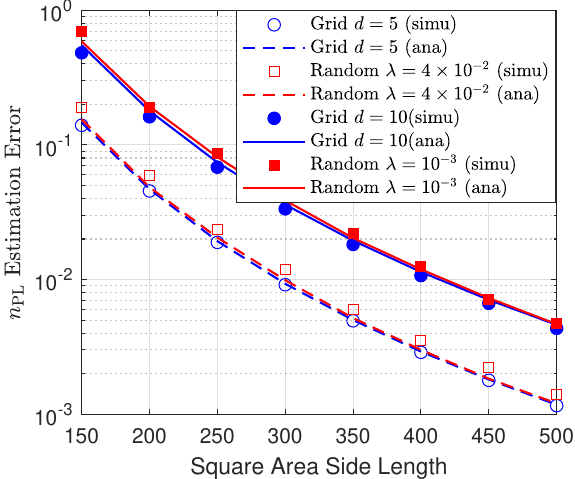}
\caption{$\beta=1$}
\end{subfigure}
\begin{subfigure}{0.33\textwidth}
\centering
\includegraphics[width=\textwidth]{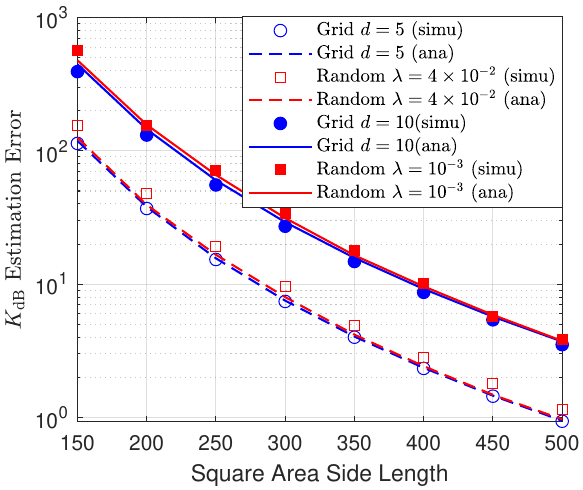}
\caption{$\beta=1$}
\end{subfigure}
\begin{subfigure}{0.33\textwidth}
\centering
\includegraphics[width=\textwidth]{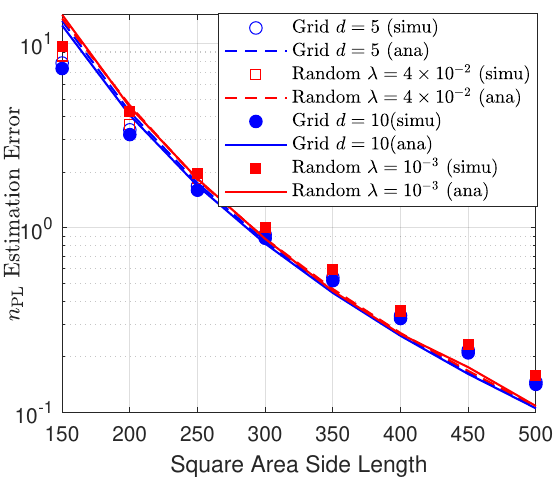}
\caption{$\beta=30$}
\end{subfigure}
\begin{subfigure}{0.33\textwidth}
\centering
\includegraphics[width=\textwidth]{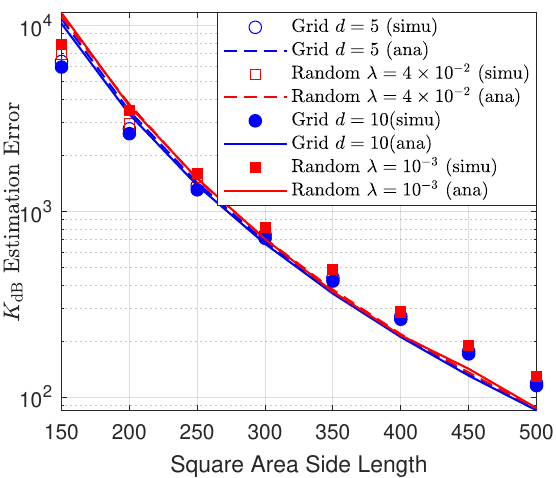}
\caption{$\beta=30$}
\end{subfigure}
\caption{The estimation error of channel parameters $n_{\mathrm{PL}}$ and $K_{\mathrm{dB}}$, with $\alpha=8$ and $\sigma^2=2$.}
\label{F:paraEsti}
\end{figure*}

Further consider the AMSE of channel prediction with the estimated parameters within a local model-consistent region. When $\beta$ is small, the AMSE is mainly caused by the error in path loss parameters and it changes with the number of collected data points within the region. Fig.~\ref{F:AMSE_estiPara}(a) shows that the simulation results match well with the analytical expression shown in \eqref{eq:AMSElowBeta}. The gap between the AMSE under estimated channel parameters and that under known parameters vanishes when the number of measurements within the region is large. Fig.~\ref{F:AMSE_estiPara}(b) shows the AMSE of channel prediction when $\beta=30$ in a small region with $N=20$ measurement samples.  When the region is small, the correlated shadowing may be considered as part of the path loss. This explains why the AMSE of the channel prediction with estimated parameters is lower than that of the prediction with true channel parameters when $k\leq 2$. When $k$ gets large, the shadowing loss can be properly calculated and hence the AMSE performance converges. The performance degradation caused by parameter estimation error is almost negligible, and hence the performance can still be approximated by the analytical results given in Lemma~\ref{lem:AMSEGrid}.

\begin{figure}[htb]
\centering
\begin{subfigure}{0.36\textwidth}
\centering
\includegraphics[width=\textwidth]{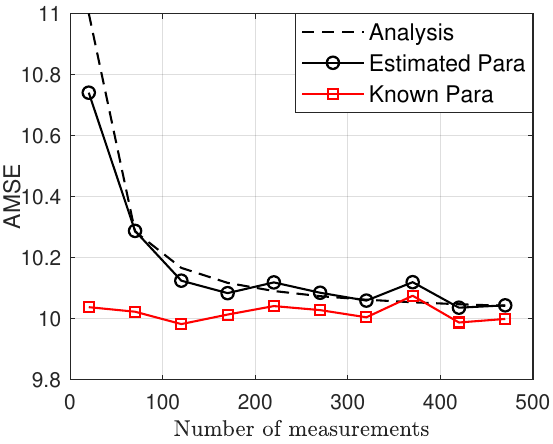}
\caption{$\beta=1,k=0$}
\end{subfigure}
\begin{subfigure}{0.36\textwidth}
\centering
\includegraphics[width=\textwidth]{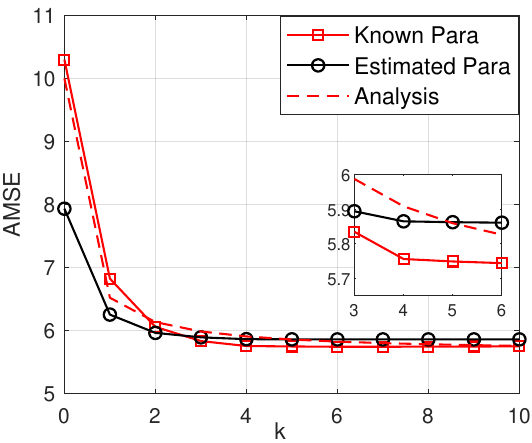}
\caption{$\beta=30,N=20$}
\end{subfigure}
\caption{The AMSE of channel prediction within the model-consistent region based on grid CGM with $d=25$.}
\label{F:AMSE_estiPara}
\end{figure}

In practical environment, finding the correct channel gain model becomes even more challenging. We consider a case study by reconstructing the CGM shown in Fig.~\ref{F:Map}(b) via the grid samples shown in Fig.~\ref{F:MapSample}(a) with $d=20$m, where the building areas are excluded. For simplicity, the whole area is only divided into two regions based on whether the LoS link is blocked by buildings or not. Following the procedures in Section~\ref{sec:estiPara}, the channel parameters for each region are estimated using the LS methods. For LoS region, the estimated channel parameters are $\hat{n}_{\mathrm{PL}}=1.3442$, $\hat K_{\mathrm{dB}}=-77.9267$, $\hat\alpha=0$, $\hat\beta=0$ and $\hat{\sigma}^2=9.4255$. For NLoS region, the estimated channel parameters are $\hat{n}_{\mathrm{PL}}=4.92$, $\hat K_{\mathrm{dB}}=-31.6323$, $\hat\alpha=339.5941$, $\hat\beta=33.9911$ and $\hat{\sigma}^2=0$. For LoS region, only the path loss need to be calculated and hence the neighboring measurements are not used for online channel prediction, i.e., with $k=0$. For NLoS region, based on the analytical results, $k=3$ is selected for balancing the accuracy and complexity for online channel prediction. The reconstructed CGM based on the estimated parameters is shown in Fig.~\ref{F:MapSample}(b) with the AMSE being 109.9952.
\begin{figure}[htb]
\centering
\begin{subfigure}{0.36\textwidth}
\centering
\includegraphics[width=\textwidth]{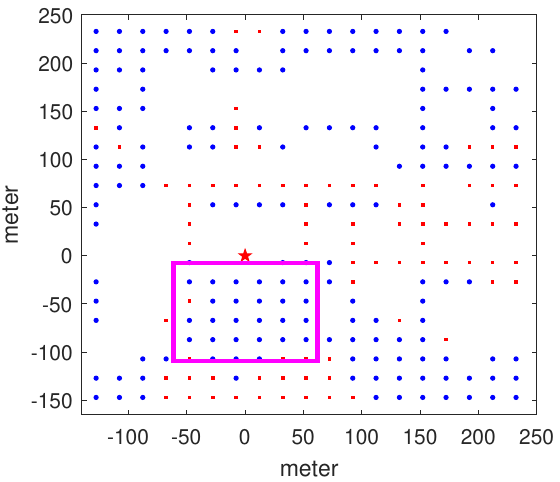}
\caption{Sample Locations}
\end{subfigure}
\begin{subfigure}{0.40\textwidth}
\centering
\includegraphics[width=\textwidth]{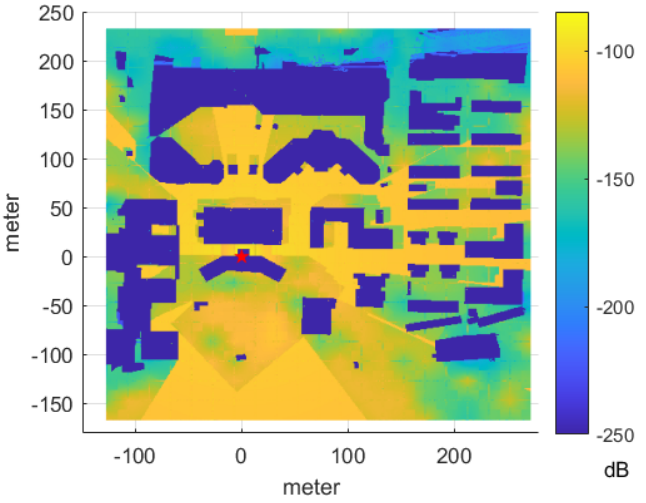}
\caption{Reconstructed CGM}
\end{subfigure}
\caption{The AMSE of channel prediction based on grid CGM with $d=20$ m with region division based on LoS link status.}
\label{F:MapSample}
\end{figure}

Further consider the channel prediction in a small region highlighted in Fig.~\ref{F:MapSample}, which contains 27 measurement samples. The estimated channel model from the samples within this region has the parameters $\hat{n}_{\mathrm{PL}}=0$, $\hat K_{\mathrm{dB}}=-120.8396$, $\hat\alpha=148.0263$, $\hat\beta=7.5880$ and $\hat{\sigma}^2=0$, and the AMSE is  67.4203, which is much smaller than the AMSE if only the LoS/NLoS regions are considered. This justifies the necessity for using different models for different regions. However, if we further divide the this region into even smaller regions, as shown in Fig.~\ref{F:SamllRegion}(a), the AMSE increases since the number of data points within each region reduces. From Fig.~\ref{F:SamllRegion}(b), it is observed that if this region is further divided into two regions (separated by the $y$-axis in Fig.~\ref{F:SamllRegion}(a)), the AMSE slightly increases to 71.2606. If this region is divided into 4 small regions, the average AMSE increases 91.1041 and it varies in different sub-regions.

\begin{figure}[htb]
\centering
\begin{subfigure}{0.36\textwidth}
\centering
\includegraphics[width=\textwidth]{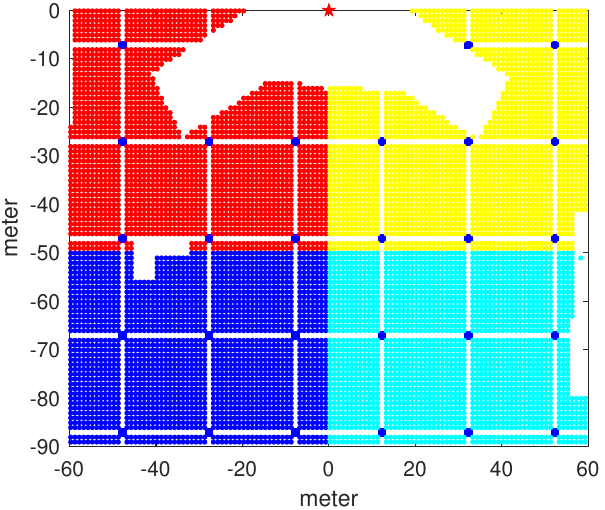}
\caption{Sample Locations}
\end{subfigure}
\begin{subfigure}{0.36\textwidth}
\centering
\includegraphics[width=\textwidth]{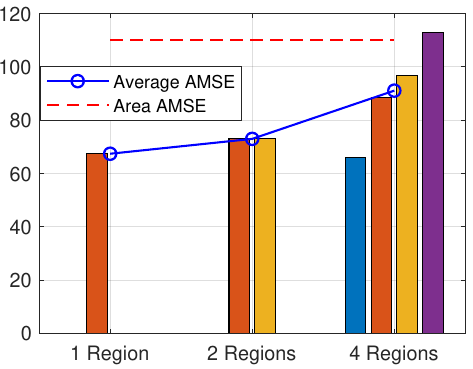}
\caption{Reconstructed CGM}
\end{subfigure}
\caption{The AMSE of channel prediction within small regions.}
\label{F:SamllRegion}
\end{figure}

\section{Conclusion}\label{sec:con}
This paper investigates the CGM construction and utilization problems towards environment-aware communications. For both the random and grid distribution of data collection locations, we derive the AMSE of channel prediction as functions of the data collection density and the number of data points used for online  prediction. To model
the spatial variation of the wireless environment, we divide the area into small regions and estimate the local channel modeling parameters based on data points within the region. The estimation errors of the path loss parameters are derived as functions of the number of samples within the region. In general, the channel prediction error reduces with the sample density and the number of measurements used for online channel prediction, so does the AMSE. Some important results from this study are summarized below:
\begin{itemize}
\item{Given the channel modeling parameters, the number of measurements need to be used for online channel estimation is far below the total number of samples within the shadowing distance, i.e., a small $k$ is sufficient in most of the scenarios. }
\item{When the correlation is strong, increasing the sample density within the region does not necessarily lead to more accurate estimation of the channel parameters, but it improves the channel prediction performance.  }
\item{For small $k$ and large $\beta$, a larger parameter estimation error may not lead to larger channel prediction error since it is unnecessary to distinguish the correlated shadowing and path loss intercept. }
\end{itemize}

Through this study, we have noticed that the proper region division is essential for model-based channel prediction. In the future, the physical map assisted region division will be considered. Besides, the data driven CGM construction is also promising in the complex urban environment, since it does not rely on the accuracy of channel model.

\appendices
%\section{Optimal Value for the State $(N,w,d)$}\label{A:finalValue}
\section{Proof of Lemma~\ref{lem:AMSEk1PPP}}\label{A:PPPk1}
Since the MSE of the estimation is only related with the distance between the target location $\mathbf q$ and its nearest data point, we can obtain the average MSE by taking the average with respect to the distribution of $d_{\min}$ in \eqref{eq:dmin_random}, which renders
\begin{align}
\mathbb{E}[\xi_{\mathrm{dB}}(\mathbf q)|_{(k=1)}]&=\alpha+\sigma^2-\frac{\alpha^2}{\alpha+\sigma^2}\int_{0}^{\infty}e^{-\frac{2x}{\beta}}2\pi\lambda x e^{-\pi\lambda x^2} dx\nonumber\\
&=\alpha+\sigma^2-\frac{2\pi\lambda\alpha^2}{\alpha+\sigma^2}\underbrace{\int_{0}^{\infty}x e^{-\pi\lambda x^2-\frac{2x}{\beta}}dx}_{\Xi(\lambda)}.
\end{align}
To further simplify the above equation, we have
\begin{align}
&\Xi(\lambda)=\int_{0}^{\infty}x\exp\left(-\pi\lambda\left(x+\frac{1}{\pi\lambda\beta}\right)^2+\frac{1}{\pi\lambda\beta^2}\right)dx\nonumber\\
&=\exp\left(\frac{1}{\pi\lambda\beta^2}\right)\left(\int_{\frac{1}{\pi\lambda\beta}}^{\infty}xe^{-\pi\lambda x^2}dx
-\frac{1}{\pi\lambda\beta}\int_{\frac{1}{\pi\lambda\beta}}^{\infty}e^{-\pi\lambda x^2}dx\right)\nonumber\\
&\overset{(a)}{=}\frac{1}{2\pi\lambda}\left(1-\frac{1}{\beta\sqrt{\lambda}}\exp\left(\frac{1}{\pi\lambda\beta^2}\right)\right)\left(1-\Phi\left(\frac{1}{\beta
\sqrt{\pi\lambda}}\right)\right)
\end{align}
where $(a)$ follows from the integration formulas in \cite{InteFormula}.
\section{Derivation of $d_{\mathrm{min}}$ Distribution in Grid CGM}\label{A:griddmin}
As shown in Fig.~\ref{F:Griddmin}(a), the middle data point is the closest neighbor for all the target locations within the red box. Hence, finding the distribution of $d_{\min}$ in the grid CGM, i.e., $P_{g}(x)=\Pr(d_{\min}=x)$, is equivalent to finding the probability that a random location within the box has a distance $x$ to the center.
\begin{figure}[htb]
\centering
\begin{subfigure}{0.20\textwidth}
\centering
\includegraphics[width=0.8\textwidth]{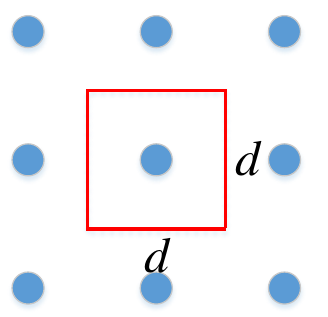}
\caption{}
\end{subfigure}
\begin{subfigure}{0.20\textwidth}
\centering
\includegraphics[width=0.8\textwidth]{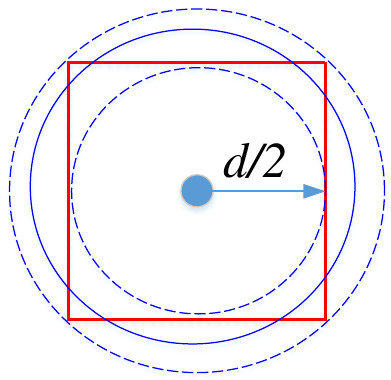}
\caption{}
\end{subfigure}
\caption{The distribution of $d_{\min}$ in Grid CGM.}
\label{F:Griddmin}
\end{figure}

From Fig.~\ref{F:Griddmin}(b), for $x\leq\frac{d}{2}$, the random location has distance $x$ with the center if it falls into the ring with inner and outer radius $x$ and $x+\epsilon$, respectively, where $\epsilon\rightarrow 0$. Hence, we have
 \begin{align}
 P_{g}(x)=\lim_{\epsilon\rightarrow 0}\frac{\pi(x+\epsilon)^2-\pi x^2}{\epsilon d^2}=\frac{2\pi x}{d^2}. \label{eq:Griddmin1}
 \end{align}

Similarly, for $\frac{d}{2}<x<\frac{\sqrt{2}d}{2}$, the random location has distance $x$ from the center if it falls into the overlapped area of the ring from $x$ and $x+\epsilon$ and the red square. Hence, we have
\begin{align}
P_{g}(x)&=\lim_{\epsilon\rightarrow 0}\frac{4x\epsilon\left(\frac{\pi}{2}-2\arccos(\frac{d}{2x})\right)}{\epsilon d^2}\nonumber\\
&=\frac{4x}{d^2}\left(\frac{\pi}{2}-2\arccos\left(\frac{d}{2x}\right)\right).
\end{align}

\section{Proof of Lemma~\ref{lem:estiError}}\label{A:ParaEstiError}
The channel parameter estimation error is given by the diagonal elements of the error covariance matrix $C_{LS}$ in \eqref{eq:paraMMSE}. Deriving the explicit expression of $C_{LS}$ is challenging due to the inverse of a $N\times N$ matrix $(R_{Q}+\sigma^2\mathbf{I}_{N\times N})$. The shadowing correlation among the data points introduces the bias error on the path loss parameter estimations. To get some insights on the parameter estimation error with CGM construction and utilization, we propose to group the data points that are within the correlation distance $\beta$ and treat each group as an effective data point, as shown in Fig.~\ref{F:PointGroup}(a). The shadowing correlation among the effective data points shown in Fig.~\ref{F:PointGroup}(b) can be neglected, i.e., with $R_{Q}=\alpha\mathbf{I}_{N\times N}$.

\begin{figure}[htb]
\centering
\begin{subfigure}{0.23\textwidth}
\centering
\includegraphics[width=\textwidth]{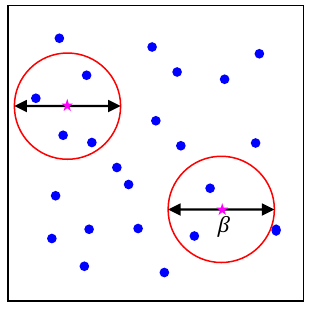}
\caption{}
\end{subfigure}
\begin{subfigure}{0.23\textwidth}
\centering
\includegraphics[width=\textwidth]{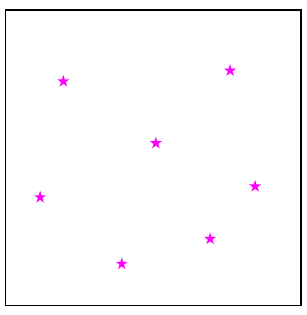}
\caption{}
\end{subfigure}
\caption{An illustration diagram for deriving the parameter estimation error.}
\label{F:PointGroup}
\end{figure}

Consider the random CGM with density $\lambda$ and grid CGM with separation $d$. The above procedures will lead to a sparse map by a factor $\pi\beta^2$, i.e., on average the number of measurements grouped together is
\begin{align}
c=\begin{cases}\max\{1,\pi\lambda\beta^2\}, & \textnormal{ Random CGM}\\
\max\{1,\pi\beta^2/d^2\}, & \textnormal{ Grid CGM}\\
\end{cases}
\end{align}

Since the multipath fading at data collection locations is assumed to be independent, the grouped data points have an effective multipath  variance $\sigma^2/c_r$. Now, we can obtain an approximation on the parameter estimation error by considering the estimation over the effective data points after grouping, which renders
\begin{align}
&C_{LS}\approx (\alpha+\sigma^2/c)(\tilde{H}^T\tilde{H})^{-1} \nonumber\\
&=\left(\alpha+\frac{\sigma^2}{c}\right)\left[\begin{matrix}\tilde{N} & -10\log_{10}\prod_{i=1}^{\tilde{N}}d_i\\
-10\log_{10}\prod_{i=1}^{\tilde{N}}d_i & \sum_{n=1}^{\tilde{N}}(10\log_{10}d_i)^2\end{matrix}\right]^{-1},\label{eq:reducedCovarance}
\end{align}
where $\tilde{H}$ contains the distance information of the effective data points, $\tilde{N}=N/c$ is the number of effective data points, $d_i=\|\mathbf q_i\|$ is the distance between the $i$th effective data location to the BS at the origin. We have assumed that $d_i$ is a random variable uniformly distributed within the range $[\delta_{\min},\delta_{\max}]$, then its value in dB scale, denoted by $y_i=10\log_{10}d_i$, is distributed according to
\begin{align}
f(y)=\Pr(y_i=y)=\frac{\ln(10)}{(\delta_{\max}-\delta_{\min})}10^{\left(\frac{y}{10}-1\right)}.
\end{align}

Denote by $\mu=\frac{1}{\tilde{N}}\sum_{i=1}^{\tilde{N}}y_i$. When $\tilde{N}$ is sufficiently large, $\mu$ will approach the expectation of $y_i$, which is derived as
\begin{align}
\mu&=\int_{10\log_{10}\delta_{\min}}^{10\log_{10}\delta_{\max}}yf(y)dy\nonumber \\
&=\frac{10\delta_{\max}\log_{10}(\delta_{\max})-10\delta_{\min}\log_{10}(\delta_{\min})}{\delta_{\max}-\delta_{\min}}-\frac{10}{\ln10}
\end{align}

Denoted by $\chi$ the variance of the random variable $y_i$, we have
\begin{align}
\chi&=\int_{10\log_{10}\delta_{\min}}^{10\log_{10}\delta_{\max}}(y-\mu)^2f(y)dy\nonumber\\
&=\frac{100}{(\ln10)^2}-\frac{100\delta_{\max}\delta_{\min}\log_{10}(\delta_{\max}/\delta_{\min})}{(\delta_{\max}-\delta_{\min})^2}.
\end{align}

When $\tilde{N}$ is sufficiently large, we have $\sum_{n=1}^{N}(10\log_{10}d_i)^2=\sum_{n=1}^{\tilde{N}}y_i^2=\tilde{N}(\chi+\mu^2)$. Hence the error covariance matrix can be in \eqref{eq:reducedCovarance} can be written as
\begin{align}
C_{LS}&=(\alpha+\sigma^2/c)\left[
\begin{matrix}\tilde{N} & -\tilde{N}\mu \\ -\tilde{N}\mu  & \tilde{N}(\chi+\mu^2)\end{matrix}
\right]^{-1}\nonumber\\
&=\frac{\alpha+\sigma^2/c}{\tilde{N}\chi}\left[
\begin{matrix}\chi+\mu^2 & \mu \\ \mu  & 1\end{matrix}
\right]
\end{align}
Together with \eqref{eq:ParaestiError}, it completes the proof of Lemma~\ref{lem:estiError}.

%
%which is approximated by $N$ times the second order statistics of the random variable $y$. When $N$ is sufficiently large, we have
%\begin{align}
%&\frac{1}{N}\sum_{n=1}^{N}y_i^2=\int_{10\log_{10}\delta_{\min}}^{10\log_{10}\delta_{\max}}y^2f(y)dy\nonumber \\
%&=\frac{\delta_{\max}(10\log_{10}(\delta_{\max}))^2-\delta_{\min}(10\log_{10}(\delta_{\min}))^2}{\delta_{\max}-\delta_{\min}}\nonumber\\
%&-\frac{200\delta_{\max}\log_{10}(\delta_{\max})-200\delta_{\min}\log_{10}(\delta_{\min})}{(\delta_{\max}-\delta_{\min})\ln(10)}+\frac{200}{(\ln(10))^2}
%\end{align}

\bibliographystyle{ieeetr}
\bibliography{CKM}

\end{document}